\documentclass[11pt]{article} 
\usepackage[utf8]{inputenc}
\usepackage[a4paper]{geometry}
\usepackage{amsfonts,amsmath,amssymb,amsthm} 
\usepackage{latexsym}
\usepackage{graphics,graphicx,epsfig}
\usepackage[normalsize,normal,sc]{caption}
\usepackage{algorithm}
\usepackage{algorithmic}

\def \no{\noindent}
\def \dis{\displaystyle}
\def \l{\left}
\def \r{\right}
\def \B{\Big}
\def \beq{\begin{equation}}
\def \eeq{\end{equation}}
\def \bal{\begin{align}}
\def \eal{\end{align}}
\def \ben{\begin{eqnarray*}}
\def \een{\end{eqnarray*}}
\def \m{\medskip}
\def \b{\bigskip}
\def \s{\smallskip}

\def \E{\mathcal{E}}
\def \F{\mathcal{F}}
\def \H{\mathcal{H}}
\def \O{\mathcal{O}}
\def \P{\mathcal{P}}
\def \V{\mathcal{V}}

\newtheorem{defi}{Definition}

\newtheorem{rem}{Remark}
\newtheorem{thm}{Theorem}
\newtheorem{cor}{Corollary}
\newcommand{\bc}{\begin{center}}
\newcommand{\ec}{\end{center}}

\begin{document}

\title{\bf A note on Pr\"{u}fer-like coding and counting forests of uniform hypertrees} 
\author{Christian Lavault\thanks{LIPN (UMR CNRS 7030), Universit\'{e} Paris~13 99, av. J.-B. Cl\'{e}ment 
93430 Villetaneuse (France). \texttt{E-mail: lavault@lipn.univ-paris13.fr}}
}
\date{\empty}
\maketitle

\vskip -.5cm
\begin{abstract}
This note presents an encoding and a decoding algorithms for a forest of (labelled) rooted uniform hypertrees 
and hypercycles in linear time, by using as few as $n - 2$ integers in the range $[1,n]$. 
It is a simple extension of the classical Pr\"{u}fer code for (labelled) rooted trees to an encoding for forests 
of (labelled) rooted uniform hypertrees and hypercycles, which allows to count them up according to their number 
of vertices, hyperedges and hypertrees. In passing, we also find Cayley's formula for the number of (labelled) rooted 
trees as well as its generalisation to the number of hypercycles found by Selivanov in the early 70's.

\s \no \textbf{Key words:} Hypergraph, Forest of (labelled rooted) hypertrees, Pr\"{u}fer code, Encoding-decoding, 
$b$-uniform, Enumeration.
\end{abstract}

\section{Notations and definitions}
A \textit{hypergraph} $\H$ is a pair $\H = (\V,\E)$, where $\V = \{1,2,\,\ldots,n\}$ denotes the set 
of vertices and $\E$ is a family of subsets of $\V$ each of size $\ge 2$ called hyperedges (see e.g.~\cite{Ber73}). 

\no Two vertices are \textit{neighbours} if they belong to the same hyperedge. The \textit{degree} of a vertex 
is the number of its neighbours. A \textit{leaf} is a set of $b - 1$ non-distinguished vertices of degree 
$(b - 1)$ belonging to the same hyperedge.

\no A \textit{hyperpath} (path) between two vertices $u$ and $v$ is a finite sequence 
of hyperedges $e_1,\,\ldots,\,e_k$, such that $e_i\, \cap\, e_{i+1} \neq \emptyset$ for any $1\le i\le k-1$, 
with $u \in e_1$ and $v \in e_k$. 

\m A hypergraph $\H$ is \textit{connected} if there exists a path between any two vertices of $\H$. A connected 
hypergraph is also called a \textit{connected component}, or simply a \textit{component}. 

A hypergraph is called \textit{b-uniform} (or uniform) if every hyperedge $e\in \E$ contains exactly $b$ vertices 
($2\le b\le n$)~\cite{KaL97,RaR06}. For example, 2-uniform hypergraphs are simply graphs. In the present note, 
only connected $b$-uniform hypergraphs ($2\le b\le n$) are considered. \\

\m The \textit{excess} of a connected hypergraph $\H = (\V,\E)$ is defined as 
$$exc(\H) = \sum_{e\in \E} (|e| - 1) - |\V|$$ 
(see e.g.~\cite{KaL97,RaR06}). Thus, if $\H$ is $b$-uniform, its excess is $(b - 1)|\E| - |\V|$. A \textit{hypertree} 
is a component of excess $-1$, which is the smallest excess possible for a connected hypergraph $\H$, 
and hence, $exc(\H)\ge -1$ for any $\H$ (see the above definition). 

\no A component is \textit{rooted} if one of its vertices is distinguished from all others. A hypergraph is called 
a \textit{forest} if all its components have excess $-1$ (i.e. are hypertrees), and similarly a \textit{hypercycle} 
has excess 0.

\section{Bijective enumeration of a forest of hypertrees} 

\subsection{State of the art and motivations}
Concerning connected graphs (2-uniform hypergraphs), there exist several methods for counting trees 
(see e.g.~\cite{CFP07,Cay89,ChW00,Joy81,Knu73,Lab81,Moo67}), including of course the Pr\"{u}fer code. Pr\"{u}fer sequences 
were first introduced by Heinz Pr\"{u}fer to prove Cayley's enumeration formula in 1918~\cite{Pru18}. 
In his very elegant proof~\footnote{Pr\"{u}fer's Theorem is as follows. \textit{There are $n^{n-2}$ sequences 
(called Pr\"{u}fer sequence or Pr\"{u}fer codes) of length $n-2$ with entries being from natural numbers; we establish 
a bijection between the set of trees and this set of sequences.}}, Pr\"{u}fer verified Caley's Theorem~\cite{Cay89} 
by establishing a one-to-one correspondence between labelled free trees of order $n$ and all sequences of $n - 2$ 
positive integers from $1$ to $n$. The Pr\"{u}fer codes can thus be generated by a simple iterative algorithm 
(see also~\cite[vol.~1, chap.~2]{Knu73} and~\cite{FlS09}).
\begin{rem}
To compute the Pr\"{u}fer sequence \textit{Seq(T)} for a labelled tree $T$, iteratively delete the leaf with smallest label 
and append the label of its neighbor to the sequence. After $n - 2$ iterations a single edge remains and we have 
produced a sequence \textit{Seq(T)} of length $n - 2$.
\end{rem}
Since the introduction of Pr\"{u}fer codes, a linear time algorithm for its computation was given 
for the first time only in the 70’s~\cite{NiW78,Ren70}, and has been later rediscovered several times 
in various forms~\cite{CFP07,ChW00,EdG01}. Recently for example, the sequential encoding and decoding schemes 
presented in~\cite{CFP07}. Both require an optimal $\Theta(n)$ time when applied to rooted $n$-node trees, 
and provide the first optimal linear time decoding algorithm for Neville's codes~\cite{CFP07}.

\b Amongst the most recent results, ``Pr\"{u}fer-like'' encoding-decoding algorithms are generalizing the Pr\"{u}fer 
code to the case of hypertrees (uniform or arbitrary). In 2009, S. Shannigrahi, S.P. Pal have shown in~\cite{ShP09} 
that uniform hypertrees can be Pr\"{u}fer-like encoded (and decoded) in optimal linear time $\Theta(n)$, using only 
$n - 2$ integers in the range $[1,n]$ (see Pr\"{u}fer's Theorem in~\cite{Pru18}). 

In the case when hypertrees are arbitrary (non uniform), the same authors' encoding and decoding algorithms in~\cite{ShP09} 
require codes of length $(n - 2) + p $, where $p$ is the number of vertices belonging to more than one hyperedge, 
or \textit{pivots}. Since the number $p$ of pivots is bounded by $|\E|$, at most $(n - 2) + |\E|$ integers are needed 
to encode a general hypertree. Therefore, the design of efficient Pr\"{u}fer-like encoding and decoding algorithms 
can be extended to arbitrary hypertrees where each hyperedge has at least two vertices. Up until now, no better 
bounds are known for the length of Pr\"{u}fer-like codes for arbitrary hypertrees. By contrast, the exact number 
of distinct hypertrees is known to be $\dis \sum_{i=0}^{n-1} \l\{{n-1\atop i}\r\} n^{i-1}$, where 
$\dis \l\{{p\atop q}\r\}$ denotes the Stirling numbers of the second kind~\cite{Knu73}.

\b The main motivation of the present note comes from analytic and bijective combinatorics of hypergraphs, including 
the enumeration of (labelled) rooted forests of uniform hypertrees and hypercycles~\cite{RaR06}. These enumeration results 
are actually tightly linked to Pr\"{u}fer-like coding and decoding of such combinatorial structures. 

The following two algorithms code and decode forests of rooted uniform hypertrees, which in turns allows 
to enumerate these structures bijectively by using a generalization of the Pr\"{u}fer sequences. The knowledge of 
the number of forests of rooted uniform hypertrees provides a concise and simple description of the structures, 
e.g. by using a recursive pruning of the leaves in the forests.

\subsection{Encoding and decoding algorithms for a forest of (labelled) rooted hypertrees}

\begin{defi} \label{def}
The forests $\F$ composed of $(k + 1)$ (labelled) rooted b- uniform hypertrees, with $n$ vertices 
and $s$ hyperedges, is coded with a 4-tuple $(R,r,\P,N)$ defined as follows:
  \begin{itemize} \setlength{\itemsep}{-.1cm} 
	\item $R$ is a set of $(k + 1)$ vertices (roots) with distinct labels in $[n]\equiv \{1,\ldots,n\}$,
	\item $r\in R$ is one the $(k + 1)$ roots,
	\item $\P$ is a partition of $[n]\setminus R$ into $s$ subsets, each of size $(b - 1)$ and
	\item $N$ is an $(s - 1)$-tuple in $[n]^{s-1}$. 
  \end{itemize}
\end{defi}
\no Note that the above positive integers $n\equiv n(s)$, $s$ and $k$ meet the condition 
$$n = s(b - 1) + k + 1,$$ 
and since $\F$ is $b$-uniform, $exc(\F) = s(b - 1) - n$. So, $|\E| = s$ and $|\V| = n$.

\b Algorithm 1 is coding a given forest $\F$ of $k + 1$ (labelled) rooted uniform hypertrees as input, 
and returns the 4-tuple $(R,r,\P,N)$ coding $\F$ as output.

\algsetup{linenosize=\small,linenodelimiter=. 
} 
\begin{algorithm}[t]
  \caption{Encoding a forest of rooted hypertrees.}
  \label{alg1}
  \begin{algorithmic}[1]
    \REQUIRE {\normalsize A forest $\F$ of $k + 1$ (labelled) rooted $b$-uniform hypertrees, with $s$ hyperedges 
    and $n = s(b-1) + k + 1$ vertices.
    \ENSURE  The coding $(R,r,\P,N)$ of $\F$ as in Definition~\ref{def}.
    \STATE $(R,r,\P,N) \leftarrow (\{\mbox{root}\},r,\{\,\},(\,))$
    \REPEAT
    \STATE Add the set of vertices corresponding to the smallest leaf (with respect to the lexicographical order) to the partition $\P$.
    \STATE Put the vertex linking that set into the $(s-1)$-tuple $N$.
    \STATE Take $\F$ as the ``new'' forest not having the vertices corresponding to the smallest leaf.
    \UNTIL {there is no hyperedge remaining in $\F$.}
    \STATE The last vertex in $N$ is necessarily a root and $r$ is redefined as this last vertex.
    \RETURN {$(R,r,\P,N)$.}
    }
  \end{algorithmic}
\end{algorithm}
\no The 4-tuple $(R,r,\P,N)$ is obtained from the coding in Definition~\ref{def} of a forest $\F$ (Algorithm~1) as follows. 
  \begin{enumerate} \setlength{\itemsep}{-.1cm}
	\item the number of its components (i.e. the number $|R|$ of its roots),
	\item the unique root vertex $r$ attached to the leaf of the last hyperedge in the pruning process,
	\item the number of hyperedges $|N| + 1$, and finally, 
	\item the number of hypertrees that are not reduced to their roots, namely the number of distinct roots in the pair $(N,r)$.
  \end{enumerate}
\no \textbf{Example}. \\
The forest $\F = (\V,\E)$ of rooted uniform hypertrees depicted in Fig.~\ref{fig1} is as follows: \\
$\V = \{1, \,2, \,\ldots, \,22\}$ and \\
$\E = \{ \{1,21,22\}, \{2,17,18\}, \{3,13,19\} \{4,8,18\}, \{4,12,14\}, \{6,7,13\}, \{7,20,21\}, \\
\{10,13,15\}, \{11,18,21\} \}$. \\
The roots of the hypertrees are $5, \,9, \,13$ and $16$.

\b \begin{figure}[ht]
    \centering
\includegraphics[width=0.9\textwidth]{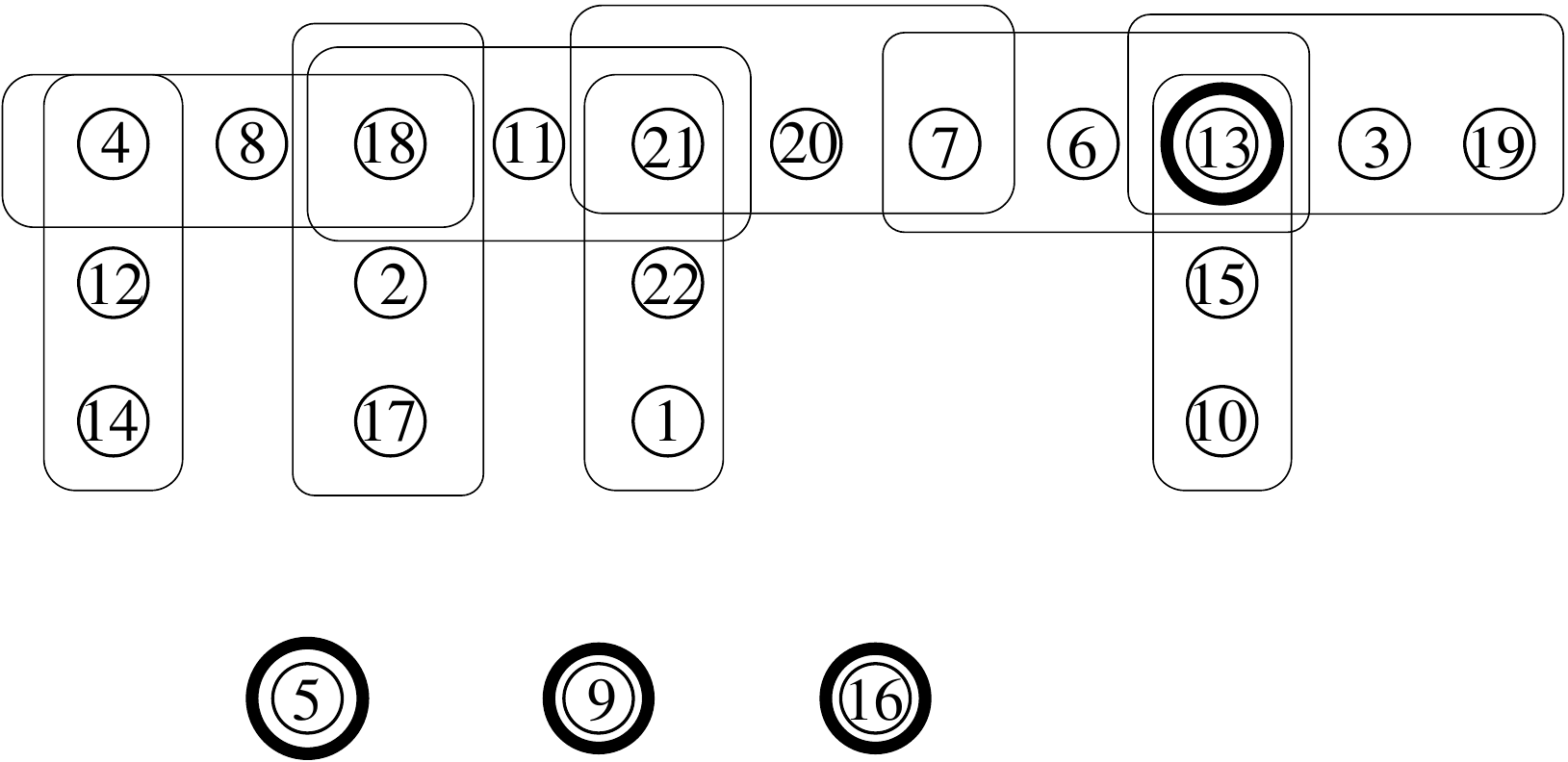}
\m  \caption{A forest of $4$ (labelled) rooted hypertrees.}
    \label{fig1}
\end{figure}

\s \no In this example, Algorithm 1 outputs $(R,\,r,\,\P,\,N)$ s.t. \\
$R = \{5,9,13,16\}$, \\
$r = 13$, \\
$\P = \{ \{1,22\},\{2,17\},\{3,19\},\{4,8\},\{6,7\},\{10,15\},\{11,18\},\{12,14\},\{20,21\} \}$ and \\
$N = (21,18,13,13,4,18,21,7)$.

\b Next, the following Algorithm 2 decodes a given 4-tuple $(R,r,\P,N)$ as input, and returns a forest 
of (labelled) rooted $b$-uniform hypertrees as output.

\algsetup{linenosize=\small,linenodelimiter=. 
}
 \begin{algorithm}[ht]
  \caption{Decoding to a forest of rooted hypertrees.}
  \label{alg2}
  \begin{algorithmic}[1]
    \REQUIRE {\normalsize Integers $n, k$ and $s$ meeting the condition  $n = s(b-1) + k + 1$  
    and the coding $(R,r,\P,N)$ of $\F$, as in Definition~\ref{def}.
    \ENSURE  Forest $\F$ of rooted $b$-uniform hypertrees, whose $k + 1$ roots are the vertices of $R$.
    \REPEAT
    \STATE Build one hyperedge with the first vertex in the $(s-1)$-tuple $N$ and the vertices %
    of the first set of $\P$ (w.r. to the lexicographic order) having no vertex still in the remaining set $N$.
    \STATE Remove the above set from $\P$ and delete the first vertex from the $(s-1)$-tuple $N$. 
    \UNTIL {the set $N$ is empty.}
    \STATE Build the $(s - 1)$-th hyperedge with the last subset of $\P$ and the vertex $r$.
    \RETURN {the forest $\F$ obtained.}
    }
  \end{algorithmic}
\end{algorithm}
\no Within the loop of Algorithm~\ref{alg2}, the forest $\F$ is found with no ambiguity by choosing 
the smallest leaf in the lexicographical order.

\begin{rem} 
Encoding Algorithm~1 and decoding Algorithm~2 are both running in optimal linear time, by using as few as $n - 2$ 
integers in the range $[1,n]$. This complexity result is a direct consequence of the Pr\"{u}fer-like encoding algorithm 
designed in~\cite{ShP09} (see Pr\"{u}fer's Theorem in~\cite{Pru18}). The algorithm computes the \textrm{hyperedge partial 
order} on the hyperedges of the hypertree in linear time by defining a directed acyclic graph (DAG) with vertex set $\E$, 
where each vertex represents a hyperedge of $\H$. (The proof is completed in~\cite[Lemma~2]{ShP09}.)
\end{rem}

\subsection{Enumeration of forests, hypertrees and trees} \label{enum}
From Algorithm~1, we obtain the enumeration of forests with $(k + 1)$ (labelled) rooted uniform hypertrees and $s$ hyperedges.

\begin{thm} \label{forests}
The number of forests having $(k + 1)$ (labelled) rooted b-uniform hypertrees and s hyperedges is
\beq \label{eq:forests} 
{n\choose k+1}\,(k+1)\,\l[\frac{(n - k - 1)!}{s!\,(b-1)!^{s}}\r]\,n^{s-1} %
\,=\; \frac{n!}{k!}\,\frac{n^{s-1}}{s!\,(b-1)!^{s}}\,, 
\eeq
where the number of vertices is $n \equiv n(s) = s(b - 1) + k + 1$.
\end{thm}
\begin{proof}
The proof stems directly from the one-to-one correspondence constructed in Algorithm~1, which codes forests $\F$ 
with the set of 4-tuples $(R,r,\P,N)$ defined as in Definition~\ref{def}. Indeed, the number of forests 
of $(k + 1)$ (labelled) rooted $b$-uniform hypertrees and $s$ hyperedges is equal to $|R|\times \# roots\times |\P|\times |N|$. 

Now, we have $\dis |R| = {n\choose k+1}$, the number of roots is $(k +1)$ , $\dis |\P| = \frac{(n - k - 1)!}{s!\,(b-1)!^{s}}$ 
and $|N| = n^{s-1}$. So, after simplifications, Theorem~\ref{forests} follows.
\end{proof}

Setting $k = 0$ in the above Eq.~(\ref{eq:forests}) (Theorem~\ref{forests}) yields $\dis n!\, \frac{n^{s-1}}{s!(b-1)!^{s}}$.
Whence the following

\begin{cor} \label{hypertrees}
The number of (labelled) rooted b-uniform hypertrees with s hyperedges is
$$\frac{(n-1)!}{s!\,(b-1)!^{s}}\,n^{s},$$
where the number of vertices is $n \equiv n(s) = s(b - 1) + 1$.
\end{cor}
\no Note that, whenever $b = 2$ (and thus $s = n - 1$), Corollary~\ref{hypertrees} is a generalization of Cayley's 
Theorem enumerating rooted trees of $n$ vertices: \textit{for any $n\ge 1$, the number of (nonplane labelled) 
rooted trees of n vertices is} $n^{n-1}$~\cite{Cay89}.

\m One of the advantages offered by a bijective construction proof is also the possibility of performing 
a random generation and learn some characteristic properties of the structures in $\F$~\cite{RaR06}. 

\b A generalization of Subsection~\ref{enum} also gives an explicit expression of the number 
of uniform hypercycles and obtain an alternative proof of Selivanov's 1972 enumeration result. 

\section{Enumeration of uniform hypercycles}
Together with hypertrees, hypercycles are the simplest structures and they have excess 0. 
Along the same lines as in Theorem~\ref{forests}, the following bijective proof gives the enumeration formula 
(first given by Selivanov in~\cite{Sel72}) of the uniform hypercycles in a forest $\F$. 

\begin{thm} \label{hypercycles} \mbox{\cite{Sel72}}\ 
The number of  b-uniform hypercycles with s hyperedges is
$$\l(\frac{(b-1) n!\,n^{s-1}}{2\,(b-1)!^{s}}\r)\; \sum_{j=2}^{s} \frac{j}{s^{j} (s-j)!} %
\;=\; \l(\frac{(b-1) n!\,n^{s-1}}{2\,(b-1)!^{s}}\r)\; \frac{1}{s(s-2)!}\,,$$
where the number of vertices is $n \equiv n(s) = s(b - 1)$.
\end{thm}
\begin{proof}
A hypercycle having a cycle length $j$ corresponds to a forest $\F$ of $j(b - 1)$ rooted hypertrees, up to an arrangment 
of the hypertrees in all distinct ways of forming a cycle. $\F$ has $(s - j)$ hyperedges and $j(b - 1)$ components 
to be arranged into a cycle. The number of $b$-uniform hypercycles having a cycle length $j$ ($2\le j\le s$) 
and with $s$ hyperedges is thus
\beq \label{eq:hypercycles}
\l[ {n\choose j(b-1)}j(b-1)\,\frac{\B((s-j)(b-1)\B)!}{(s-j)!\,(b-1)!^{s-j}} \r]\; %
\l[ \frac{1}{2}\, \frac{\B(j(b-1)\B)!}{(b-2)!^{j}} \r],
\eeq
where $n \equiv n(s) = s(b - 1)$. 

\s \no In the above Eq.~(\ref{eq:hypercycles}) indeed, the left-hand factor (between brackets) counts the number 
of forests of rooted $b$-uniform hypertrees, while the right-hand one counts the number of smooth hypercycles 
labelled with the set $\{1,\ldots,j(b-1)\}$. $j$ distinct hyperedges, and thus labelled with the set $\{1,\ldots,j(b-1)\}$. 

\s Now, $(b-1)!^{-s+j} (b-2)!^{-j} = (b-1)!^{j} (b-1)!^{-s}$ and, since $n = s(b-1)$, we have $(s-j)(b-1) = n - j(b-1)$.
So, Eq.~(\ref{eq:hypercycles}) simplifies to $\dis \frac{n!(b-1)}{2\,(b-1)!^{s}}\; \frac{(b-1)!^{j}}{(s-j)!}\,$,
with $j$ ranging from 2 to $s$.

\s Finally, substituting $n^{s-1}/s^{j}$ for $(b-1)!^{j}$ and summing on $2\le j\le s$, gives the finite sum 
in Theorem~\ref{hypercycles}: $\dis \sum_{j=2}^{s} \frac{j}{s^{j} (s-j)!} \,=\, \frac{1}{s(s-2)!}\,$,
and the result follows.
\end{proof} 

\section{Conclusions and further results}
In the above proof of Theorem~\ref{hypercycles} we are led to distinghish hypercycles according to the lengths 
of the cycles. Therefore, a question arises: for a given number $n = s(b-1)$ of vertices, what is the cycle 
length $j$ of the class that contributes most to the number of such hypercycles?

It is shown in~\cite{ShP09} that there exists at most $\dis \frac{n^{n-2} - f(n,b)}{(b-1)^{(b-2)(n-1)/(b-1)}}$ 
distinct labelled $b$-uniform hypertrees, where $f(n,b)$ is a lower bound on the number of labelled trees of maximal 
(vertex) degree exceeding $\dis \Delta = (b-1) + \frac{n-1}{b-1} - 2$. In view of extending this result, can we determine 
a lower bound on the number of labelled trees with no constraint on their maximal (vertex) degree---or, at least, 
with maximal degree exceeding some $\Delta' < \Delta$? This, for example, by designing some generalized counting 
techniques based on a bijective or analytic enumeration of $b$-uniform hypertrees. (See also~\cite{Ren70,ShP09}). 

\m In the spirit of~\cite{FlS09}, some potential applications of Pr\"{u}fer-like code may also arise as fruitful 
directions of research.
 
Encoding algorithms (such as the present one or the algorithm designed in~\cite{ShP09}) can be used to generate 
unique identities (IDs) or PINs. By generating distinct hypertrees with combinatorial enumeration methods, 
it is possible to compute distinct codes for each of these hypertrees using such encoding algorithms. 
All codes generated this way would be distinct and can provide unique IDs. The advantage of the scheme 
is that no check for repetitions is needed, since IDs generated from distinct hypertrees are unique. 
Besides, the generation of such codes requires time proportional to the length of the code. 

Coding schemes can also be useful for allocating IDs to different users in a system where disjoint sets of users 
form different groups. Each group is associated with a distinct hypertree, whereas the users within a group are 
allocated distinct codes of the same hypertree associated with the group. Subgroups can be realized by another 
level of Pr\"{u}fer-like encoding. The actual implementation of such group management schemes is an open direction 
of research (see~\cite{ShP09}).

\end{document}